\documentclass{CSML}

\def\dOi{12(2:11)2016}
\lmcsheading%
{\dOi}
{1--8}
{}
{}
{Nov.~19, 2015}
{Jun.~29, 2016}
{}

\ACMCCS{[{\bf Theory of computation}]: Models of
  computation---Concurrency---Process calculi} 
\keywords{bisimulation, up-to techniques, coinduction, companion}

\usepackage[active]{srcltx}
\usepackage[latin1]{inputenc}
\usepackage{amssymb}
\usepackage{amsmath}

\usepackage{amsthm}
\usepackage{stmaryrd}
\usepackage{hyperref}
\usepackage{versions}

\theoremstyle{definition}

\newcommand{\act}{a}
\newcommand{\states}{{\bf S}}
\newcommand{\trans}[2]{{#1} \stackrel{\act}{\longrightarrow} {#2}}
\newcommand{\bisim}{\stackrel{\mbox{\bf .}}{\sim}}
\newcommand{\progress}{\rightarrowtail}
\newcommand{\F}{\mathcal{F}}
\newcommand{\LRF}{\mbox{\sc{lrf}}}  
\newcommand{\R}{\mathrel{R}}
\renewcommand{\S}{\mathrel{S}}

\begin{document}
\title[The Largest Respectful Function]{The Largest Respectful Function}
\author[J.~Parrow]{Joachim Parrow}
\author[T.~Weber]{Tjark Weber}
\address{Uppsala University, Sweden}

\begin{abstract}
  Respectful functions were introduced by Sangiorgi as a compositional
  tool to formulate short and clear bisimulation proofs.  Usually, the larger the respectful function, the easier the bisimulation proof.
  In particular the largest respectful function, defined as the
  pointwise union of all respectful functions, has been shown to be very useful.  We here
  provide an explicit and constructive characterization of it.
\end{abstract}

\maketitle

\section{Introduction}

\subsection{Bisimulation and up-to techniques}

The well known method of bisimilarity for defining behavioural
equivalence on labelled transition systems works as follows. A
symmetric binary relation $\R$ is a {\em bisimulation} if whenever $p
\R q$ and $\trans{p}{p'}$ then $\trans{q}{q'}$ and $p'\R q'$.  In
other words, whatever $p$ can do can be mimicked by $q$ such that the
derivatives are still related.  The idea is usually attributed to
Park~\cite{Park1981} although similar notions in logics and non well
founded sets were present earlier, and it was popularized by Milner in
his subsequent papers and book~\cite{Milner1989} on communicating
systems.  Bisimulations are closed under union and therefore the union
of them all is the largest bisimulation, written~$\bisim$ and called
{\em bisimilarity}. The main point is that if $p \bisim q$ then $p$
and $q$ can mimic each other indefinitely, and are thus inseparable
for an observer who can only detect the labels of the transitions.

In order to establish $p \bisim q$ we must find a bisimulation
relation containing the pair $(p,q)$. On the one hand we would like
this relation to be as small as possible, since for every pair in it
we must check that all transitions can be mimicked. On the other hand
we would also want it as large as possible in order to facilitate the
proof of $p'\R q'$, that the derivatives are related. These apparently
conflicting interests were noted already by Milner~\cite{Milner1989}
who suggested a remedy: instead of requiring $p'\R q'$ in the
consequent it suffices to require $p' \,\mathop{\bisim} \circ R \circ
\mathop{\bisim}\, q'$. This means that before establishing membership
in~$R$ we are allowed to replace the derivatives with already known
bisimilar ones. In many cases this makes the proofs significantly
easier and clearer. Milner dubbed the technique {\em bisimulation
  up-to~$\bisim$} and it quickly caught on. Variants of it were used
in many other process algebras, and other up-to techniques such as
up-to one-hole contexts
turned out to be useful.

The first to establish a systematic theory of up-to techniques was
Sangiorgi~\cite{Sangiorgi1994,Sangiorgi1998}, defining several
important notions. One is a {\em progress} relation between binary
relations. Briefly put the relation~$R$ progresses to the
relation~$S$, written $R \progress S$, if whenever $p \R q$ and
$\trans{p}{p'}$ then $\trans{q}{q'}$ and $p'\S q'$, and vice versa. A bisimulation is
thus a relation that progresses to itself. Milner's original up-to
technique uses a relation $\R$ progressing to $\mathop{\bisim} \circ R
\circ \mathop{\bisim}$. In general an up-to technique uses a function
$\F$ on relations such that $\R$ progresses to $\F(R)$, i.e., if $p \R
q$ and $\trans{p}{p'}$ then $\trans{q}{q'}$ and $p' \mathrel{\F(R)}
q'$. The question then is which such $\F$ are {\em sound} in the sense
that they guarantee that $R$ is included in bisimilarity, and how they
can be constructed in a systematic way.

\subsection{Respectful Functions}

Sangiorgi showed that although the sound functions in general are not
closed under composition, a subset of them, called the {\em
  respectful} functions, are closed under composition, (pointwise)
union and iteration. This means that complicated respectful functions
can be constructed in a modular way from simple components, which
clarifies and shortens bisimulation proofs: the effect is that up-to
techniques based on respectful functions can be freely combined. The
respectfulness criterion is that if $R \subseteq S$ and $R\progress S$
then $\F(R) \subseteq \F(S)$ and $\F(R) \progress \F(S)$. In other
words, if a relation progresses to a superset then the same must hold
when $\F$ is applied. Similar notions have lent structure to proofs in
more advanced settings such as the pi-calculus~\cite{Hirschkoff1997}
and psi-calculi~\cite{Pohjola2016}. The idea has also been recast in a
more general form by Pous~\cite{Pous2007}, where the corresponding
notion is of so called {\em compatible} functions, though we shall
here stick with Sangiorgi's original definitions and notations.

In a typical bisimulation proof one begins with processes $p$ and $q$
to be proved bisimilar, and then defines a relation $R$ containing
them such that $R$ progresses to $\F(R)$ for some respectful $\F$. It
then follows that $R \subseteq \mathop{\bisim}$ and the proof is
concluded. It does not matter which $\F$ is used as long as it is
respectful. Often,  the larger $\F(R)$ the easier the proof. Therefore
a viable strategy is to simply choose $\F(R)$ as the largest of all
possible images of $R$ for a respectful function. This is well defined
since the respectful functions are closed under arbitrary pointwise
union; we can simply define the largest respectful function as the
pointwise union of all respectful functions. In other words a relation
$R$ is mapped to the union of all $\F(R)$ for all respectful $\F$.

To our knowledge the first explicit mention of something similar to
this is in Hur et al.~\cite{Hur2013} where the largest respectful
function is denoted by~$\dagger$. It is remarked that~$\dagger$ is
useful in coinductive proofs even though it lacks a constructive
definition. The authors write ``\ldots the greatest respectful up-to
function is so powerful, we see no point in ever stating a proof
component's contribution involving a different respectful up-to
function.''  In recent work~\cite{Pous2016} Pous presents a similar
idea, that of the largest compatible function, called the {\em
  companion}, and demonstrates its usefulness in a variety of
settings.  In a sense it is remarkable that although the notion of
respectfulness has been around for more than 20 years, the largest and
therefore most general respectful function has not been investigated
and given an independent characterization.

The contribution of this short note is to give an explicit
characterization of the largest respectful function. It is easy to see
that if $R \subseteq \mathop{\bisim}$ then it must map $R$ to
$\bisim$, but for other $R$ the situation has been less clear. Our
characterization uses Milner's stratification of
bisimilarity~\cite{Milner1989} $\sim_\alpha$ for any ordinal
$\alpha$. In brief our result is that the largest respectful function
maps $R$ to $\sim_\alpha$ where $\alpha$ is the largest ordinal such
that $R \subseteq \,\sim_\alpha$.
The following section provides detailed definitions and proof of this
result. In the final section we briefly mention a generalisation to
complete lattices and demonstrate the connection with the compatible
functions by Pous.

\section{Result}

We assume a labeled transition system~$(\states, \Delta,
\longrightarrow)$ where ${\bf S}$ is a set of states, $\Delta$ is a
set of labels, and $\longrightarrow \subseteq \states \times \Delta
\times \states$. We let $p$, $q$ etc.\ range over $\states$ and $\act$
over $\Delta$. For $(p, a, p') \in \mathop{\longrightarrow}$ we write
$\trans{p}{p'}$. We also let $R$, $S$, $X$ range over binary relations
on~$\states$.

\subsection{Progression and Respectfulness}
\begin{defi}[Progress]
\label{progress}
We say that~$R$ progresses to~$S$, denoted $R \progress S$, if for all
$(p,q) \in R$:
\begin{enumerate}
\item $\forall p',\act.\ \trans{p}{p'} \Longrightarrow \exists
  q'.\ \trans{q}{q'} \wedge (p',q') \in S$, and
\item $\forall q',\act.\ \trans{q}{q'} \Longrightarrow \exists
  p'.\ \trans{p}{p'} \wedge (p',q') \in S$
\end{enumerate}
\end{defi}

The connection with bisimilarity is that a bisimulation is a relation
that progresses to itself, and $\bisim$ is the union of all
bisimulations, but we shall formally not need these notions in our
proofs below.

\begin{lem}
\label{progress-monotone}
For all relations~$R'$, $R$, $S$, $S'$,
$$R' \subseteq R \wedge S \subseteq S' \wedge R \progress S
  \Longrightarrow R' \progress S'$$
\end{lem}

\begin{proof}
Immediate from Definition~\ref{progress}.
\end{proof}

\begin{defi}[Respectfulness]
\label{respectfulness}
A function~$\F$ on relations is called \emph{respectful} if, for all
relations~$R$ and~$S$,
$$R \subseteq S \wedge R \progress S \Longrightarrow \F(R) \subseteq
\F(S) \wedge \F(R) \progress \F(S)$$
\end{defi}

It is known~\cite{Sangiorgi1998} that the pointwise union of
arbitrarily many respectful functions is respectful.  Thus, there is a
largest respectful function, namely the pointwise union of all
respectful functions.  Our main result below is to provide an
alternative characterization of this function.

\subsection{Stratification of Bisimilarity}

We will use Greek letters $\alpha$, $\beta$, \ldots\ to range over the
ordinals, and reserve $\lambda$ to stand for a limit ordinal. The
following definition is due to Milner~\cite{Milner1989}, originally
with the intention of connecting bisimilarity to logical formulas.

\begin{defi}[Stratifications of Bisimilarity]
\label{def:stratifications}
For every ordinal~$\alpha$, define a relation~$\sim_\alpha$ as
follows:
\begin{align*}
\mathop{\sim_0}        &:=\ {\bf S} \times {\bf S}\\
\mathop{\sim_{\alpha+1}} &:=\ \bigcup \{X \mid X \progress \mathop{\sim_\alpha}\}\\
\sim_\lambda            &:=\ \bigcap_{\alpha < \lambda} \mathop{\sim_\alpha}
\end{align*}
\end{defi}

The first clause can be regarded as a special case of the clause for
limit ordinals (an intersection of zero sets is the universal
relation), but it is clearer to write out the base case explicitly.
Milner's idea is that as $\alpha$ increases, $\sim_\alpha$ corresponds
to finer behavioural equivalences.  Thus $p \sim_0 q$ always holds, $p
\sim_1 q$ means that $p$ and $q$ can mimic each other for one
transition, $p \sim_2 q$ that they can mimic each other for two
transitions, and so on.  If $\longrightarrow$ is finitely branching,
i.e., for each state the set of outgoing transitions is finite, this
sequence converges at $\omega$, i.e., $\mathop{\sim_\omega} =
\mathop{\sim_{\omega+1}}$.  In general this is not the case.  The
construction in the following lemma is illuminating, and we have not
seen it explicitly stated before, although we shall not need it for
our main result.

\begin{lem}
\label{distinct-equivalences}
For any ordinal~$\alpha$ there exists a transition system
where~$\sim_\alpha$ and~$\sim_{\alpha+1}$ are distinct.
\end{lem}

\begin{proof}
We will only need transition systems with a single transition label,
which we elide.  For any ordinal~$\alpha$ let the transition
system~$T_\alpha$ contain as states all ordinals less than or equal
to~$\alpha$, and let the transitions be defined by ordinal membership,
i.e., $\alpha \longrightarrow \beta$ if $\alpha > \beta$. Since the
transitions from a state will be the same in all~$T_\alpha$ where the
state occurs, we may elide explicit references to the transition
systems below.  We can now prove the following for all ordinals
$\alpha, \beta, \gamma$ with $\alpha < \beta$:
\[\mbox{$\alpha \sim_\gamma \beta$ iff $\gamma \leq \alpha$}\]
The proof is by transfinite induction over $\gamma$.  The cases where
$\gamma=0$ or a limit are immediate. Assume $\gamma = \gamma' +
1$. For the direction $\Leftarrow$, the only transitions from~$\alpha$
or~$\beta$ that cannot be mimicked directly (i.e., leading to exactly
the same state) are $\beta \longrightarrow \beta'$ for some $\beta'
\geq \alpha$.  A simulating transition is then $\alpha \longrightarrow
\gamma'$ and by induction $\gamma' \sim_{\gamma'} \beta'$.
Conversely, assume $\alpha \sim_{\gamma'+1} \beta$ and consider the
transition $\beta \longrightarrow \alpha$.  Since $\alpha$ can
simulate there is a transition $\alpha \longrightarrow \alpha'$ such
that $\alpha' \ \sim_{\gamma'} \alpha$ . By induction $\gamma' \leq
\alpha'$ whence $\gamma \leq \alpha$.

To conclude the proof of the lemma, take $T_{\alpha+1}$ where it holds
that $\alpha \sim_\alpha \alpha+1$ and $\alpha \not\sim_{\alpha+1}
\alpha+1$.
\end{proof}

We now establish that the relations $\sim_\alpha$ indeed become
smaller as $\alpha$ increases:

\begin{lem}
\label{stratifications-successor}
For all ordinals~$\alpha$, $\mathop{\sim_{\alpha+1}} \subseteq
\mathop{\sim_\alpha}$.
\end{lem}

\begin{proof}
By transfinite induction over~$\alpha$.
\begin{enumerate}
\item $\alpha = 0$: Trivially $\mathop{\sim_1} \subseteq {\bf S}
  \times {\bf S} = \mathop{\sim_0}$.
\item $\alpha = \beta+1$: The induction hypothesis entails
  $\mathop{\sim_{\beta+1}} \subseteq \mathop{\sim_\beta}$.  Hence for
  any relation~$X$, $X \progress \mathop{\sim_{\beta+1}}$ implies $X
  \progress \mathop{\sim_\beta}$ by Lemma~\ref{progress-monotone}.
  Thus $\mathop{\sim_{\alpha+1}} = \bigcup\{X \mid X \progress
  \mathop{\sim_{\beta+1}}\} \subseteq \bigcup \{X \mid X \progress
  \mathop{\sim_\beta}\} = \mathop{\sim_{\beta+1}} =
  \mathop{\sim_\alpha}$.
\item $\alpha = \lambda$: For all~$\beta < \lambda$,
  $\mathop{\sim_\lambda} \subseteq \mathop{\sim_\beta}$ is immediate
  from Definition~\ref{def:stratifications}.  As in the successor
  case, this implies $\mathop{\sim_{\lambda+1}} \subseteq
  \mathop{\sim_{\beta+1}}$.  Moreover, $\mathop{\sim_{\beta+1}}
  \subseteq \mathop{\sim_\beta}$ by the induction hypothesis.  Thus
  $\mathop{\sim_{\lambda+1}} \subseteq \mathop{\sim_\beta}$ for
  all~$\beta < \lambda$, and therefore $\mathop{\sim_{\lambda+1}}
  \subseteq \bigcap_{\beta < \lambda} \mathop{\sim_\beta} =
  \mathop{\sim_\lambda}$.\qedhere
\end{enumerate}
\end{proof}

\begin{lem}
\label{stratifications-subset}
For all ordinals~$\alpha$ and~$\beta$, $\alpha \leq \beta
\Longrightarrow \mathop{\sim_\beta} \subseteq \mathop{\sim_\alpha}$.
\end{lem}

\begin{proof}
By transfinite induction over~$\beta$, using
Lemma~\ref{stratifications-successor} for the successor case.  The
limit case is immediate from Definition~\ref{def:stratifications}.
\end{proof}

\subsection{Convergence of the Stratification}

As demonstrated in Lemma~\ref{distinct-equivalences}, there is no
universal ordinal to which the series of equivalences~$\sim_\alpha$
converges in all transition systems.  For the rest of this paper we
assume some fixed transition system.  The following lemmas then
establish that no matter what this transition system is, there exists
an ordinal where the series converges.

\begin{lem}
\label{stratifications-varepsilon}
There exists an ordinal~$\varepsilon$ such that
$\mathop{\sim_{\varepsilon+1}} = \mathop{\sim_\varepsilon}$.
\end{lem}

The lemma is an instance of a well-known fixed point result for
monotone functions.  Our proof follows Rubin and
Rubin~\cite{Rubin1963}.

\begin{proof}
By Hartogs's theorem~\cite{Hartogs1915} we can find an
ordinal~$\kappa$ such that there is no injection from~$\kappa$ into
the powerset of ${\bf S} \times {\bf S}$.  Therefore, there exist
ordinals $\varepsilon < \beta < \kappa$ such that
$\mathop{\sim_{\varepsilon}} = \mathop{\sim_\beta}$.
Lemma~\ref{stratifications-subset} then implies
$\mathop{\sim_{\varepsilon+1}} = \mathop{\sim_\varepsilon}$.
\end{proof}

In the following we write $\varepsilon$ for the least ordinal provided
by Lemma~\ref{stratifications-varepsilon}, i.e., such that
$\mathop{\sim_{\varepsilon+1}} = \mathop{\sim_\varepsilon}$.

\begin{lem}
\label{stratifications-stable}
For all ordinals~$\alpha$, $\alpha \geq \varepsilon \Longrightarrow
\mathop{\sim_\alpha} = \mathop{\sim_\varepsilon}$.
\end{lem}

\begin{proof}
By transfinite induction over~$\alpha$, using
Lemma~\ref{stratifications-varepsilon} for the successor case.  The
limit case follows from the induction hypothesis and
Lemma~\ref{stratifications-subset}.
\end{proof}

\subsection{Progression of Strata}

We next show that $\sim_\varepsilon$ progresses to itself.

\begin{lem}
\label{stratifications-progress}
For all ordinals~$\alpha$, $\mathop{\sim_{\alpha+1}} \progress
\mathop{\sim_\alpha}$.
\end{lem}

\begin{proof}
Suppose $p \sim_{\alpha+1} q$.  By
Definition~\ref{def:stratifications}, $(p,q) \in X$ for some~$X$ such
that $X \progress \mathop{\sim_\alpha}$.
\end{proof}

\begin{lem}
\label{stratifications-bisimilarity}
$\mathop{\sim_\varepsilon} \progress \mathop{\sim_\varepsilon}$
\end{lem}

\begin{proof}
From Lemmas~\ref{stratifications-progress}
and~\ref{stratifications-varepsilon}.
\end{proof}

Lemma~\ref{stratifications-bisimilarity} proves
that~$\mathop{\sim_\varepsilon}$ is included in bisimilarity, which is
defined as the union of all relations progressing to themselves.  The
opposite inclusion, which we will not need in this note, is a
straightforward exercise (again using transfinite induction
over~$\alpha$).

\subsection{The Largest Respectful Function}

We can now give an explicit characterization of the largest respectful
function:

\begin{defi}[$\LRF$]
\label{definition-LRF}
$$\LRF(R) := \bigcap \{\mathop{\sim_\alpha} \mid R \subseteq
  \mathop{\sim_\alpha}\}$$
\end{defi}

In other words, by Lemma~\ref{stratifications-subset}, $\LRF(R)$ is
$\sim_\alpha$ for the smallest $\sim_\alpha$ containing all of $R$. In
particular, if $R \subseteq\, \sim_\varepsilon$ then $\LRF(R) =\,
\sim_\varepsilon$. The crucial properties of $\LRF$ are:

\begin{lem}
\label{monotone}
$\LRF$ is monotone, i.e., for all relations~$R$ and~$S$, $$R \subseteq
S \Longrightarrow \LRF(R) \subseteq \LRF(S)$$
\end{lem}

\proof
Suppose~$R \subseteq S$.  Then $S \subseteq \mathop{\sim_\alpha}$
implies $R \subseteq \mathop{\sim_\alpha}$.  Hence $$\LRF(R) = \bigcap
\{\mathop{\sim_\alpha} \mid R \subseteq \mathop{\sim_\alpha}\}
\subseteq \bigcap \{\mathop{\sim_\alpha} \mid S \subseteq
\mathop{\sim_\alpha}\} = \LRF(S)
\rlap{\hbox to 77 pt{\hfill\qEd}}$$
\newpage

\begin{thm}
\label{lrf-is-respectful}
$\LRF$ is respectful.
\end{thm}

\begin{proof}
Suppose~$R \subseteq S$ and~$R \progress S$.  Then~$\LRF(R) \subseteq
\LRF(S)$ by Lemma~\ref{monotone}.  It remains to show $\LRF(R)
\progress \LRF(S)$.  We consider two cases.

Case~1: $S \subseteq \mathop{\sim_\alpha}$ for all ordinals~$\alpha$.
Then in particular $S \subseteq \mathop{\sim_\varepsilon}$.  Thus also
$R \subseteq \mathop{\sim_\varepsilon}$, and $\LRF(R) =
\mathop{\sim_\varepsilon}\ \progress\ \mathop{\sim_\varepsilon} =
\LRF(S)$ by Lemma~\ref{stratifications-bisimilarity}.

Case~2: There exists an ordinal~$\alpha$ such that $S \not\subseteq
\mathop{\sim_\alpha}$.  Since the ordinals are well-ordered, we may
wlog.\ assume that~$\alpha$ is minimal, i.e., $S \subseteq
\mathop{\sim_\beta}$ for all $\beta < \alpha$.  Note that $\alpha \neq
0$ since $S \subseteq {\bf S} \times {\bf S} = \mathop{\sim_0}$.
Moreover,~$\alpha$ is not a limit ordinal, since $S \subseteq
\mathop{\sim_\beta}$ for all $\beta < \lambda$ implies $S \subseteq
\bigcap_{\beta < \lambda} \mathop{\sim_\beta} =
\mathop{\sim_\lambda}$.  Thus $\alpha = \gamma + 1$ for some
ordinal~$\gamma$.

From $R \progress S$ and $S \subseteq \mathop{\sim_\gamma}$ we have $R
\progress \mathop{\sim_\gamma}$ by Lemma~\ref{progress-monotone}.
Hence $R \subseteq \bigcup \{X \mid X \progress \mathop{\sim_\gamma}\}
= \mathop{\sim_{\gamma+1}}$.  Therefore $\LRF(R) = \bigcap
\{\mathop{\sim_\beta} \mid R \subseteq \mathop{\sim_\beta}\} \subseteq
\mathop{\sim_{\gamma+1}}$.

Lemma~\ref{stratifications-subset} implies $S \not\subseteq
\mathop{\sim_\beta}$ for all $\beta \geq \alpha$.  Thus $\LRF(S) =
\bigcap \{\mathop{\sim_\beta} \mid S \subseteq \mathop{\sim_\beta}\} =
\bigcap_{\beta<\alpha} \mathop{\sim_\beta} = \bigcap_{\beta \leq
  \gamma} \mathop{\sim_\beta} = \mathop{\sim_\gamma}$, again using
Lemma~\ref{stratifications-subset}.

$\LRF(R) \progress \LRF(S)$ now follows from $\mathop{\sim_{\gamma+1}}
\progress \mathop{\sim_\gamma}$ (Lemma~\ref{stratifications-progress})
and Lemma~\ref{progress-monotone}.
\end{proof}

Finally we establish that $\LRF$ is the largest respectful function in
the sense that it contains all respectful functions:

\begin{thm}
\label{lrf-is-largest}
If~$\F$ is respectful, then $\F(R) \subseteq \LRF(R)$ for every
relation~$R$.
\end{thm}

\begin{proof}
Suppose $\F$ is respectful.  We prove $$\forall R.\ R \subseteq
\mathop{\sim_\alpha} \Longrightarrow \F(R) \subseteq
\mathop{\sim_\alpha}$$ for all ordinals~$\alpha$ by transfinite
induction.  ($\F(R) \subseteq \LRF(R)$ then follows from the
definition of~$\LRF$.)

\begin{enumerate}
\item $\alpha = 0$: Trivially $\F(R) \subseteq {\bf S} \times {\bf S}
  = \mathop{\sim_0}$.
\item $\alpha = \beta + 1$: Suppose $R \subseteq
  \mathop{\sim_{\beta+1}}$.  By Lemma~\ref{stratifications-progress},
  $\mathop{\sim_{\beta+1}} \progress \mathop{\sim_\beta}$.  Thus $R
  \progress \mathop{\sim_\beta}$ by Lemma~\ref{progress-monotone}.
  Moreover, $\mathop{\sim_{\beta+1}} \subseteq \mathop{\sim_\beta}$ by
  Lemma~\ref{stratifications-successor}.  Hence $R \subseteq
  \mathop{\sim_\beta}$, and $\F(R) \progress \F(\mathop{\sim_\beta})$
  follows from Definition~\ref{respectfulness}.  The induction
  hypothesis (applied to~$\mathop{\sim_\beta}$) implies
  $\F(\mathop{\sim_\beta}) \subseteq \mathop{\sim_\beta}$.  Hence
  $\F(R) \progress \mathop{\sim_\beta}$ by
  Lemma~\ref{progress-monotone}.  Therefore $\F(R) \subseteq \bigcup
  \{X \mid X \progress \mathop{\sim_\beta}\} =
  \mathop{\sim_{\beta+1}}$.
\item $\alpha = \lambda$: Suppose $R \subseteq \mathop{\sim_\lambda} =
  \bigcap_{\beta < \lambda} \mathop{\sim_\beta}$.  Then $R \subseteq
  \mathop{\sim_\beta}$ for all $\beta < \lambda$, and the induction
  hypothesis implies $\F(R) \subseteq \mathop{\sim_\beta}$.  Hence
  $\F(R) \subseteq \bigcap_{\beta < \lambda} \mathop{\sim_\beta} =
  \mathop{\sim_\lambda}$.\qedhere
\end{enumerate}
\end{proof}

\section{Generalisation}

We have established that Definition~\ref{definition-LRF} defines the
largest respectful function.
Our constructions and proofs are quite general and transfer smoothly
to other settings.  We can for example recast our result in a setting
of complete lattices as follows.  Assume that a set~$A$ with
order~$\leq$ is a complete lattice with top element~$\top$.  Let~$a$,
$b$ range over~$A$.  For a binary relation~$R$ on~$A$, the pre-image
of~$b$, i.e., $\{a \mid a \R b\}$, is written $R\,b$.

\begin{defi}[cf.~\protect{\cite[Definition 1.16]{Pous2007}}]
\label{progression}
A {\em progression} is a binary relation $R$ on $A$ such that:
\begin{enumerate}
\item
\label{compliant}
$\mathop{\leq} \circ \R \circ \mathop{\leq} \ \subseteq\ R$
\item
\label{inbounded}
$\forall b \in A.\; \bigvee R\,b \in R\,b$
\end{enumerate}
\end{defi}

As an example, Lemma~\ref{progress-monotone} establishes that
$\progress$ satisfies condition~(\ref{compliant}).  Moreover, it is
obvious from Definition~\ref{progress} that for any relation~$S$,
$\{R\mid R \progress S\}$ is closed under arbitrary union.  Hence
$\progress$ also satisfies~(\ref{inbounded}).

\begin{defi}
\label{Rmonotone}
For a binary relation $R$ on $A$, say that a function $f\colon A
\rightarrow A$ is {\em $R$-monotone} if it is monotone with respect to
$\mathop{\leq} \cap \R$, i.e., $a \leq b \wedge a\R b \Longrightarrow
f(a) \leq f(b) \wedge f(a)\R f(b)$.
\end{defi}

For example, by Definition~\ref{respectfulness} respectfulness is
$\progress$-monotonicity for the lattice of binary relations under the
inclusion order.

\begin{defi}
\label{general-stratification}
Given a relation $R$, for every ordinal $\alpha$ define $z_\alpha \in
A$ by the following transfinite induction:
\[\begin{array}{lcl}
z_0  & := & \top \\
z_{\alpha+1} & := & \bigvee R\, z_\alpha \\
z_\lambda  & := & \bigwedge \{z_\alpha \mid \alpha < \lambda\}
\end{array}\]
\end{defi}

\begin{thm}
\label{general-theorem}
If $R$ is a progression, the unique largest $R$-monotone function is
\[\lambda x.\; \bigwedge \{z_\alpha \mid x \leq z_\alpha\}\]
\end{thm}

The proof follows the previous section closely, with an arbitrary
progression instead of~$\progress$. The one point of deviation is in
the proof of the counterpart of Lemma~\ref{stratifications-progress}
above, that $z_{\alpha+1} \R z_\alpha$. Here we use
condition~(\ref{inbounded}) of Definition~\ref{progression} to show
that $\bigvee R\, z_\alpha$, i.e., $z_{\alpha+1}$, must lie in~$R\,
z_\alpha$.

As a special case we can consider~$A$ to be the lattice of binary
relations on states ordered by inclusion, and~$R$ to be the
progression relation~$\progress$; we then recover the theorems of the
previous section.

\medskip

\noindent \emph{Respectful vs.\ compatible.}
Pous has developed a theory of up-to techniques based on
\emph{compatible} rather than respectful functions.
Definition~\ref{progression} is equivalent to Definition~1.16
in~\cite{Pous2007}.  To a progression~$R$ Pous associates the monotone
function~$s_R\colon A\to A$ given by $\lambda x.\; \bigvee R\,x$.
Conversely, given a monotone function~$s\colon A\to A$, the set
$\{(a,b) \mid a \leq s(b)\}$ defines a progression.  A monotone
function $f\colon A\to A$ is \emph{$s$-compatible} if $f \circ s \leq
s \circ f$ pointwise.

According to Proposition~1.17(ii) in the same paper, a monotone
function~$f\colon A\to A$ is $s_R$-compatible (for a progression~$R$)
iff, for all $a$, $b\in A$, $a \R b$ implies $f(a) \R f(b)$.  Note
that $R \cap \mathop{\leq}$ is a progression whenever~$R$ is a
progression.  It follows that for monotone functions, $R$-monotonicity
(Definition~\ref{Rmonotone}) is exactly $s_{R \cap
  \mathop{\leq}}$-compatibility.

We did not restrict ourselves to monotone functions in this note, but
since the largest respectful function is in fact monotone
(cf.\ Lemma~\ref{monotone}), Theorem~\ref{general-theorem} thus gives
the largest $s_{R \cap \mathop{\leq}}$-compatible function for any
progression~$R$.

A thorough analysis and comparison between respectfulness and
compatibility is in \cite{Pous2016}, Section~9.  In general,
$R$-monotonicity and $s_R$-compatibility are not equivalent.  For
monotone functions, $R$-monotonicity is strictly weaker than
$s_R$-compatibility.  However, these differences turn out to be
irrelevant when we consider the largest function, allowing us to
establish a more direct connection between respectfulness and
compatibility.

\begin{thm}
\label{theorem-companion}
If $R$ is a progression, the unique largest $s_R$-compatible function
is
\[\lambda x.\; \bigwedge \{z_\alpha \mid x \leq z_\alpha\}\]
\end{thm}\smallskip

\noindent The proof is a minor adaptation of the proof for
Theorem~\ref{general-theorem}.  In particular, to prove that $\lambda
x.\; \bigwedge \{z_\alpha \mid x \leq z_\alpha\}$ is $s_R$-compatible,
we note that $a \R b$ and $b \leq z_\alpha$ implies $a \R z_\alpha$ by
progression, hence $a \leq \bigvee R\,z_\alpha = z_{\alpha+1} \leq
z_\alpha$.  It is then straightforward to adjust the (generalized)
proofs of Theorems~\ref{lrf-is-respectful} and~\ref{lrf-is-largest} to
use $s_R$-compatibility instead of $R$-monotonicity.

As a corollary of this explicit characterization, we
recover~\cite[Proposition~9.1]{Pous2016}: the largest respectful
function and the largest compatible function coincide.

\section*{Acknowledgement}

\noindent We are grateful to Johannes \AA{}man Pohjola for pointing
out the connection with~\cite{Pous2007}, and for discussions with
Damien Pous.

\bibliographystyle{alpha}
\bibliography{bibliography}

\end{document}